\documentclass[preprint,pra,superscriptaddress]{revtex4-1}

\usepackage[dvips]{graphicx} 
\usepackage{framed}
\usepackage{amsfonts,amsmath,amssymb}    
\usepackage{amsthm}
\usepackage{enumerate}
\usepackage{epsfig}
\usepackage{diagbox}
\usepackage{subfigure}
\usepackage{xcolor}

\newtheorem{theorem}{Theorem}[section]

\newtheorem{Claim}[theorem]{Claim}

\newtheorem{corollary}[theorem]{Corollary}

\begin{document}

\title{{Improved key rate bounds for practical decoy-state quantum key distribution systems}}

\author{Zhen Zhang}
\author{Qi Zhao}
\affiliation{Center for Quantum Information, Institute for Interdisciplinary Information Sciences, Tsinghua University, Beijing, 100084, China}
\author{Mohsen Razavi}
\email{m.razavi@leeds.ac.uk}
\affiliation{School of Electronic and Electrical Engineering, University of Leeds, Leeds, LS2 9JT, UK}
\author{Xiongfeng Ma}
\email{xma@tsinghua.edu.cn}
\affiliation{Center for Quantum Information, Institute for Interdisciplinary Information Sciences, Tsinghua University, Beijing, 100084, China}

\begin{abstract}

The decoy-state scheme is the most widely implemented quantum key distribution {protocol} in practice. In order to account for the finite-size key effects on the achievable secret key generation {rate}, a rigorous statistical fluctuation analysis is required. {Originally}, a heuristic Gaussian-approximation technique was used for this purpose, which, despite of its analytical convenience, was not sufficiently rigorous. The fluctuation analysis has recently been made rigorous by using the Chernoff bound. {There} is a considerable gap, however, between the key rate bounds obtained from these new techniques and that obtained from {the} Gaussian assumption. Here, we develop a tighter bound for the decoy-state method, {which} yields a smaller failure probability. This improvement results in a higher key rate and increases the maximum distance over which {secure} key exchange is possible. By optimizing the system parameters, our simulation results show that our new method almost closes the gap between the two previously proposed techniques and achieves a similar performance to that of conventional Gaussian approximations.


\end{abstract}
\maketitle

\section{Introduction}
In theory, quantum key distribution (QKD) \cite{Bennett:BB84:1984,Ekert:QKD:1991} has been proven to be information-theoretically secure against eavesdropping attacks \cite{mayers2001unconditional,Lo1999Science,PhysRevLett.85.441}, even if we assume that the attacker, Eve, has full control over the channel. The security of QKD stems from the complementary relation of non-commuting measurement operators in quantum mechanics \cite{Koashi:Comp:09}. Due to the uncertainty principle, any Eve's interference that gains {her some} information about the key would inevitably introduce disturbance. The users, Alice and Bob, can then bound the information leakage to Eve by quantifying the disturbance. The latter requires collecting data from which certain parameters of the system, such as bit and phase error probabilities \cite{PhysRevLett.85.441}, can accurately be estimated.

In practice, the required probabilities above cannot be directly measured. Instead, one can only measure the rates, i.e., the frequencies of occurrence. If the QKD system {runs} for an infinitely long time, the rates will converge to the corresponding underlying probabilities. That is, the parameters needed for data postprocessing can be measured accurately when the data size is {sufficiently} large. In reality, there are deviations between rates and probabilities due to \textit{statistical fluctuations}. A finite-key analysis accounts for these deviations and derives a security parameter, the failure probability, for the final key. With the aid of the finite-key analysis, the security of QKD can also be extended to {its} composable security definition \cite{BenOr:Security:05,Renner:Security:05}. The finite-key analysis of QKD systems with idealized single-photon sources and detectors are well studied in the literature \cite{Ma2011Finite}. Here, we develop tight bounds for the secret key rate in practical scenarios when decoy states are in use \cite{PhysRevLett.91.057901,Lo:PRL:2005,PhysRevLett.94.230503}.

A perfect single-photon source is hard to attain in practice. Alternatively, a highly attenuated laser, described by a weak coherent state, is widely used in QKD. The multi-photon components in the coherent state would introduce security loopholes in practice \cite{duvsek1999generalized,Brassard2000PNS}. Such imperfections in realistic devices were originally taken into consideration in the Gottesman-Lo-L\"utkenhaus-Preskill (GLLP) security analysis \cite{gottesman2004security}. By directly applying the GLLP analysis to the coherent-state QKD system, however, the performance, {measured by} key rate and {maximum} secure transmission distance, is rather limited \cite{Ma2006low}. A clever twist to the weak-laser QKD, known as the decoy-state method, is introduced in \cite{PhysRevLett.91.057901,Lo:PRL:2005,PhysRevLett.94.230503}, {which, fortunately, }can enhance system performance {to} a level comparable to that of a perfect single-photon source. {The decoy-state method is now widely used in QKD systems \cite{Zhao2006DecoyExp,zhao2006simulation,Rosenberg:ExpDecoy:2007,Zeilinger:Decoy:2007,Peng:ExpDecoy:2007,YSS_Decoy_07}.}


In the decoy-state method, we estimate the channel parameters by sending two types of states. One is called {the} signal state, which is used to transmit keys similar to the single-photon source in the ideal situation. The other is called {the} decoy state, which is used to characterize the channel, by estimating the number of single-photon states traversing the channel. In the information-theoretical security proof of the decoy-state method \cite{Lo:PRL:2005}, these two states have the same properties except for their intensity, which results in distinct Poisson distributions for their photon number. Note that the phases of the coherent states must be randomized, in order that the source can be treated as a statistical mixture of Fock states. In this case, the channel, controlled by Eve, will have the same impact on the single-photon components in both signal and decoy states. The channel parameters, such as the probability of a single photon passing through, defined as the single-photon yield, would then be the same for the signal and decoy states. This property is at the core of the security of the decoy-state technique. We revisit this condition in our finite-key analysis.

Estimating the channel parameters, such as the single-photon yield, would become less accurate when one only has a finite set of data. Statistical fluctuation must then be considered, in our security analysis, to account for possible deviations from true (probability) values. It turns out that the statistical fluctuation analysis for the decoy-state method can be a complicated problem. To simplify the problem, a Gaussian distribution assumption on the channel fluctuations was made in early analyses \cite{XMA:PRA:2005}. Throughout the paper, we refer to this Gaussian approximation technique by the Gaussian analysis method. Such an assumption is not necessarily justified when one considers a rigorous security proof. Lately, this Gaussian assumption was removed from the security proof by applying the Chernoff bound and the Hoeffding inequality \cite{curty2014finite,PhysRevA.89.022307}. We refer to this latter technique by the Chernoff+Hoeffding method.

The simulation results show that a {large-size} key is required to achieve a secure key with the Chernoff+Hoeffding method and the key rate is {lower} than that of the Gaussian analysis method. In this work, we improve the finite-key analysis method and {provide} a tighter estimation {of QKD parameters} by breaking the parameter estimation problem into different regimes of operation and finding tight bounds in each case. After optimizing the system parameters, we show that our improved finite-key analysis method {achieves} a similar performance to the Gaussian analysis method.

The organization of this paper is as follows. In Sec.~\ref{Sec:Security}, we review the commonly used vacuum+weak decoy-state scheme \cite{lo2004quantum,XMA:PRA:2005} and develop a general formulation for its finite-key analysis. In Sec.~\ref{Sec:Statistical}, we present our new statistical fluctuation method, and provide instructions on how our results can be applied to a realistic experimental setup. Note that our proposed method is generic {and} can also be used in other decoy-state QKD schemes. In Sec.~\ref{Sec:Simulation}, we first construct a QKD simulation model with typical experimental parameters, and then compare our new method with previous work when each method has been optimized to offer its best performance. We discuss the results and conclude the paper in Sec.~\ref{Sec:Conclusion}.

\section{Finite-key analysis for vacuum+weak decoy-state scheme}\label{Sec:Security}
In this section, we lay out a precise formulation for our finite-key analysis problem in the special case of vacuum+weak decoy-state protocol. This turns out to offer a unifying language, applicable to both the Chernoff+Hoeffding \cite{curty2014finite,PhysRevA.89.022307} and the Gaussian analysis methods, as well as our own proposed method. We will then compare the new formulation with that of the Gaussian analysis method \cite{XMA:PRA:2005}, and show how the results there can be employed in our finite-key analysis. In particular, we show that the formulation in the Chernoff+Hoeffding method has an equivalent form to that of the Gaussian analysis method. In the following, in Sec.~\ref{section:Definitions}, we review the widely-used scheme of vacuum+weak decoy-state QKD \cite{lo2004quantum}. Then, the definitions and notations used in this paper are given. In Sec.~\ref{Sub:VMOpt}, we formulate the parameter estimation problem in its general form. Finally, in Sec.~\ref{Section:Analytical solution decoy}, we use the results in \cite{XMA:PRA:2005} to find analytical bounds for the parameters of interest.  


\subsection{Vacuum+weak decoy-state protocol}\label{section:Definitions}

The vacuum+weak decoy-state protocol, first presented in 2004 \cite{lo2004quantum}, is a widely used decoy-state scheme. In this protocol, Alice encodes the pulses with three different intensities, {corresponding} to vacuum states, weak decoy states and the signal states. This scheme is capable of estimating the single-photon components because, intuitively, when the intensity of a coherent state pulse is very weak, the resulting detection events mainly come from the single-photon components and background. The yield of the background noise can be estimated by the vacuum decoy state. By combining measurement results of weak decoy and vacuum decoy states, the relevant parameters to the single-photon components, including the yield and quantum bit error rate (QBER), can {accurately be} estimated. With those parameters, secure keys can be obtained from the signal states after postprocessing.

The protocol is described in more detail in the following steps:


\begin{enumerate}
\item
\emph{State preparation}: {For each bit in her raw key}, Alice randomly chooses {the intensity and the basis to encode her bit. She can choose from} three intensities, namely, vacuum state, weak decoy state and signal state, {and} then randomly encode {her bit} in the $X$ or $Z$ basis, and sends {it} to Bob. The probability of choosing the $Z$ basis could, in general, be different from that of the $X$ basis \cite{lo2005efficient}.

\item
\emph{Measurement}: Bob measures the received states in the $X$ {or} $Z$ basis {chosen} randomly. The probability of choosing a measurement basis is the same as that of the encoding stage.

\item
\emph{Sifting}: {Over an authenticated channel}, Alice announces the basis and signal/decoy information she has {used}, while Bob announces the locations of valid detections and the bases {used for} his measurements. {If} Alice and Bob {have chosen} the same basis, they keep the {corresponding} bits as the sifted key.

\item
\emph{Error correction and verification}: Alice calculates {some} parity information of her {sifted} key, encrypts the {parity bits} with pre-shared secure keys, and sends {them} to Bob. Bob then performs the error correction and, Alice and Bob verify if their keys are now identical \cite{Ma2011Finite}. If the verification fails they perform the error correction again or abort the protocol. If the keys are verified to be identical, Bob finds the number of bit errors and evaluates the QBER.

\item
\emph{Parameter estimation}: Using the parameters obtained in the experiment, a lower bound on the number of successful detection events resulted from {single-photon} components of the signal states, {$M_1^s$}, and an upper bound on the {corresponding} phase error rate, $e_1^{ps}$, will be obtained in each basis. The latter quantifies the leaked information to a potential eavesdropper.

\item
\emph{Privacy amplification}: Alice and Bob apply universal hashing function based on the parameters {$M_1^s$} and $e_1^{ps}$ in each basis. Then, according to the GLLP analysis \cite{gottesman2004security}, a shorter but more secure key can be extracted with a length of {$M_1^s[1-h(e_1^{ps})]$}.
\end{enumerate}

The final key length in each basis is then lower bounded by
\begin{equation} \label{VM:KeyBit}
\begin{aligned}
K &\ge M_1^{s}[1-h(e_1^{ps})]-K_{ec},\\
K_{ec} &= M^{s}fh(E^{s}),\\
\end{aligned}
\end{equation}
where $f$ denotes the inefficiency of error correction, and $h(x)=-x \log_2 x -(1-x) \log_2 (1-x)$ is the Shannon binary entropy function. Here, for the sake of simplicity, we assume that Alice and Bob only extract secure keys from the signal states. In principle, they can also extract secure keys from the decoy states as well. The other parameters in Eq.~\eqref{VM:KeyBit} are defined below.

Below, the notation used throughout the paper, including the parameters in Eq.~\eqref{VM:KeyBit}, is presented.

\begin{enumerate}
\item
The superscripts $x$ and $z$ denote the $X$ and $Z$ bases, respectively. For brevity of notation, we often do not explicitly mention the basis superscript, unless otherwise needed. All parameters defined below are then for a certain fixed basis $\gamma = x,z$, although the superscript $\gamma$ is not shown.

\item
Capital letters $K$, $N$, and $M$, respectively, denote the number of the final key bits, the pulses sent by Alice and the valid, after basis sifting, detections on Bob's side.

\item
 $Q$ denotes the gain, i.e., the rate of creating a sifted key bit, and $E$ denotes the total QBER in the sifted key bit.

\item
$Y_i$ denotes the yield of $i$-photon states, and is given by $Y_i\equiv M_i/N_i$, where the subscript $i$ for $M$ and $N$ {refers to the corresponding} counts for $i$-photon states.

\item
$e_i$ denotes the error rate corresponding to the transmission of $i$-photon states. Note that it should not be confused with the letter $e$ without the subscript, which is the base of the natural logarithm.

\item
The superscripts $s$, $w$ and $v$, respectively, denote the signal state with intensity $\mu$, weak decoy state with intensity $\nu$ ($<\mu$){, and} vacuum state. The superscript/subscript $a$ denotes these three cases, i.e., $a\in\{s,w,v\}$, with corresponding intensity $\mu_a\in\{\mu,\nu,0\}$.
\item
The superscripts $b$ and $p$ {refer to} bit and phase {(in error-rate terms)}, respectively.

\item
The superscripts $L$ and $U$ refer to the lower bound and the upper bound, respectively.

\item
$q^{a}\equiv N^a/N$ denotes the rate Alice encodes a state with intensity $\mu_a$.

\item
On Alice's side, $p_i^{a}$ denotes the conditional probability that an $i$-photon state corresponds to a coherent pulse with intensity $\mu_a$, {i.e.,}
\begin{equation} \label{VW:pimua}
\begin{aligned}
p_i^{a}&\approx\frac{N_i^{a}}{N_i},\\
\end{aligned}
\end{equation}
where the approximation is caused by statistical {fluctuations. The approximation} becomes equality in the asymptotic (infinite-key) {limit}. Due to the Poisson {distribution} of the photon numbers in different states and $N^a=q^aN$, these probabilities are given by,
\begin{equation} \label{VW:piswv0}
\begin{aligned}
p_i^{a}&=\frac{N^a e^{-\mu_a}(\mu_a)^i/i!}{\sum_{\alpha\in\{s,w,v\}}N^\alpha e^{-\mu_\alpha}(\mu_\alpha)^i/i!},\\
&=\frac{q^a e^{-\mu_a}(\mu_a)^i/i!}{\sum_{\alpha\in\{s,w,v\}}q^\alpha e^{-\mu_\alpha}(\mu_\alpha)^i/i!}.\\
\end{aligned}
\end{equation}
{Note that $p_i^{a}$ is the only probability term used in this paper. All other terms are rates, i.e., the ratio between two counts. }


\end{enumerate}

\subsection{Statistical fluctuation analysis: Formulation}
\label{Sub:VMOpt}
Our key objective in the statistical fluctuation analysis of the decoy-state schemes is to bound $M_1^s$ and $e_1^{ps}$, by allowing a certain failure rate, by using the measurement results obtained in a QKD round. A QKD round consists of transmitting $N$ pulses by Alice, out of which $K$ key bits are to be extracted. In this subsection and next, all the terms refer to the parameters in a particular basis, e.g., the $Z$ basis. The same results hold for the other basis as well. In each QKD round, Alice and Bob can specify $M^a$ and $E^{a}M^{a}$ for different values of $a$. Based on these measurement results, they consider a worst-case scenario by finding the minimum value of $M_1^s$ and the maximum value of $e_1^{ps}$ that is consistent with the measurement results.

From the GLLP security analysis \cite{gottesman2004security}, Eve cannot get any key information from the single-photon states without introducing disturbance, while she can in principle get information about the key when multiple photons are sent, say, via photon-number-splitting attacks \cite{duvsek1999generalized,Brassard2000PNS}. Eve's objective is then to minimize $M_1^s$, within the constraints of the decoy-state scheme.

Note that some parameters, such as $N_i$ and $M_i$ are, in principle, known to Eve assuming that she can perform non-demolition measurements on the signals generated by Alice. From Alice and Bob's perspective, these variables are, however, unknown, but have a fixed value in each round of the QKD protocol once Bob's measurements are completed. On the other hand, the choice of $a$ for each transmitted state is known to Alice, while Eve has no information about that before the sifting stage. This is the key advantage that Alice and Bob have over Eve in specifying the range of values that the key parameters of interest would take. In the following, we will try to find relationships between the measurable parameters $M^a$ and $E^{a}M^{a}$ and the unknown (to Alice and Bob), but fixed, parameters $M_i$. We will then show how this can help us bound $M_1^s$ and $e_1^{ps}$.

For phase-randomized coherent sources, the state prepared by Alice can be considered as a mixture of Fock states. The channel, controlled by Eve, behaves the same to different Fock states. This is called {the} photon number channel model \cite{Ma2008PhD}. For an $i$-photon state, the conditional detection probability for Bob that the originally encoded state has an intensity $\mu_a$ is the same as the probability chosen by Alice, $p_i^{a}$, defined in Eq.~\eqref{VW:pimua}. This implies that
\begin{equation} \label{VW:MimuaEM}
\begin{aligned}
M_i^{a}&\approx p_i^{a}M_i, \\
e_i^aM_i^{a}&\approx p_i^{a}e_iM_i, \\
\end{aligned}
\end{equation}
where the approximation becomes equality in the asymptotic case. 

The total number of detection events caused by the state $a$, $M^{a}$, and the number of errors, $E^{a}M^{a}$, are given by contributions from states with different numbers of photons, that is
 \begin{equation} \label{VM:MmuaEM}
\begin{aligned}
M^{a}&=\sum_i M_i^{a} \\
E^{a}M^{a}&=\sum_i e_i^{a}M_i^{a}. \\
\end{aligned}
\end{equation}
Therefore, by substituting Eq.~\eqref{VW:MimuaEM} into Eq.~\eqref{VM:MmuaEM}, we obtain
\begin{equation} \label{VM:EMlinearEq}
\begin{aligned}
M^s&\approx p_0^sM_0+\dots+p_i^sM_i +\dots,\\
M^w&\approx p_0^wM_0+\dots+p_i^wM_i +\dots,\\
M^v&\approx p_0^vM_0,\\
E^s M^s&\approx p_0^s e_0M_0+\dots+p_i^se_iM_i +\dots,\\
E^w M^w&\approx p_0^w e_0M_0+\dots+p_i^we_iM_i +\dots,\\
E^v M^v&\approx p_0^ve_0M_0, \\
\end{aligned}
\end{equation}
where the approximation becomes equality in the asymptotic case. Note that the terms on the {left hand side} of Eq.~\eqref{VM:EMlinearEq} are measurable counts, while the ones on the {right hand side} are mixed with probabilities. When the data size is finite, the statistical fluctuation may lead to deviations between $M_i^{a}$ ($e_i^{a}M_i^{a}$) and $p_i^{a}M_i$ ($p_i^{a}e_iM_i$), in Eq.~\eqref{VW:MimuaEM}, and {similarly} in Eq.~\eqref{VM:EMlinearEq}. Our objective is to bound these deviations while meeting a certain failure rate for the protocol, as we show next.

The key idea that we use to bound the right-hand side of Eq.~\eqref{VM:EMlinearEq} is to use the fact that Eve does not know the type of the states used by Alice. While Eve can control the values of $M_i$, for $i = 0, 1, 2, \dots$, she cannot change them after Bob's measurements. Nevertheless, even for fixed values of $M_i$, she cannot exactly predict the measurement results $M^a$ and $E^{a}M^{a}$. That is, before the sifting stage, these variables can be considered to be random. It turns out, however, that the expectation value of these random variables, as we show next, can be written as a weighted sum of $M_i$s. That is, after Bob's measurements, Eve can no longer change these mean values either. From Alice and Bob's point of view, a set of observed values for $M^a$ and $E^{a}M^{a}$ would correspond to a fixed, but unknown, set of values for $M_i$. Using proper techniques, they can then bound the above expectation values as a function of the observed values.

Let us first look at $M_i^a$ in a more detailed way. Before the sifting stage, but after Bob's measurements, $M_i$ has a fixed value, but $M_i^a$ is random to Eve. We can then rewrite $M_i^a$ as follows
\begin{equation}
\label{M_ia}
M_i^a = \sum_{j=1}^{M_i} {\chi_{i,j}^a},
\end{equation}
where
\begin{equation}
\chi_{i,j}^a = \left\{ \begin{array}{lr}
1& \mbox{with probability\ } p_i^a \\
0& \mbox{with probability\ } 1-p_i^a
\end{array}
\right.,
\quad\mbox{$j=1,\ldots,M_i$,}
\end{equation}
are independent and identically distributed indicator random variables. It will then follow that
\begin{equation} \label{VM:ExpectMmua}
\begin{aligned}
\mathbb{E}[M_i^{a}]&=p_i^{a}M_i,  \\
\mathbb{E}[e_i^{a}M_i^{a}]&=p_i^{a}e_iM_i, \\
\end{aligned}
\end{equation}
where $\mathbb{E}[\cdot]$ is the expectation value with respect to $\chi_{i,j}^a$ variables. Finally, from Eqs.~\eqref{VM:MmuaEM} and \eqref{VM:ExpectMmua} we find
\begin{equation} \label{VW:EME}
\begin{aligned}
\mathbb{E}[M^s]&=p_0^sM_0+\dots+p_i^sM_i +\dots, \\
\mathbb{E}[M^w]&=p_0^wM_0+\dots+p_i^wM_i +\dots, \\
\mathbb{E}[M^v]&=p_0^vM_0, \\
\mathbb{E}[E^sM^s]&=p_0^se_0M_0+\dots+p_i^se_iM_i +\dots, \\
\mathbb{E}[E^wM^w]&=p_0^we_0M_0+\dots+p_i^we_iM_i +\dots, \\
\mathbb{E}[E^vM^v]&=p_0^ve_0M_0, \\
\end{aligned}
\end{equation}
where, again, the expectation values are taken with respect to $\chi_{i,j}^a$ variables. Note that these expectation values would represent the average values for our observables from Eve's perspective before the sifting stage, but after Bob's measurements. At this stage, Alice and Bob can safely assume that Eve can no longer change the values of $M_i$ variables on the right-hand side of the above equations. The measured values for $M_a$ and $E^a M^a$ will then set some constraints on the expectation values in Eq.~\eqref{VW:EME}, and, correspondingly, the right-hand side of the above equations. In particular, we can show that for any set of values for observables $M_a$ ($E^a M^a$), we can find lower and upper bounds for their corresponding expected values, respectively, denoted by $\mathbb{E}^L[M^a]$ ($\mathbb{E}^L[E^aM^a]$) and $\mathbb{E}^U[M^a]$ ($\mathbb{E}^U[E^aM^a]$). Our finite-key analysis can then be formulated as the following optimization problem: Find
\begin{equation} \label{VW:MEMeqs}
\begin{aligned}
\min{~M_1}, \;\;\; {\rm s.t.,}\\
\mathbb{E}^L[M^s]&\le p_0^sM_0+\dots+p_i^sM_i +\dots\le \mathbb{E}^U[M^s]\\
\mathbb{E}^L[M^w]&\le p_0^wM_0+\dots+p_i^wM_i +\dots\le \mathbb{E}^U[M^w]\\
\mathbb{E}^L[M^v]&\le p_0^vM_0\le\mathbb{E}^U[M^v] \mbox{ and}\\
\max{~e_1M_1}, \;\;\; {\rm s.t.,}\\
\mathbb{E}^L[E^sM^s]&\le p_0^se_0M_0+\dots+p_i^se_iM_i +\dots\le \mathbb{E}^U[E^sM^s]\\
\mathbb{E}^L[E^wM^w]&\le p_0^we_0M_0+\dots+p_i^we_iM_i +\dots\le \mathbb{E}^U[E^wM^w]\\
\mathbb{E}^L[E^vM^v]&\le p_0^ve_0M_0\le \mathbb{E}^U[E^vM^v].\\
\end{aligned}
\end{equation}
In Sec.~\ref{Sec:Statistical}, starting with the Chernoff bound, we show how the required lower and upper bounds above can be related to the measured observables. Before doing that, however, let us find the correspondence between the above formulation and that of the previous work in \cite{XMA:PRA:2005}.

\subsection{Correspondence with Gaussian Analysis Method} \label{Section:Analytical solution decoy}
In order to compare our formulation in sec.~\ref{Sub:VMOpt} with that of the Gaussian analysis method proposed in \cite{XMA:PRA:2005}, we rewrite Eq.~\eqref{VW:EME} by dividing both sides of it by $N^a$. We obtain the following
\begin{equation} \label{QKD-ratenew:signal}
\begin{aligned}
\mathbb{E}[Q^{a}]&=\mathbb{E}[\frac{M^{a}}{N^{a}}]=\frac{\mathbb{E}[M^{a}]}{N^{a}}\\
&=\sum_{i=0}^\infty{p_i^{a}\frac{M_i}{N^{a}}}\\
&=\sum_{i=0}^\infty \frac{e^{-\mu_a}(\mu_a)^i/i!q^{a}}{e^{-\mu}\mu^i/i!q^s+e^{-\nu}\nu^i/i!q^w}\frac{M_i}{q^{a}N}\\
&=\sum_{i=0}^\infty e^{-\mu_a}\frac{(\mu_a)^i}{i!}Y_i^\ast, \\
\mathbb{E}[E^aQ^{a}]&= \sum_{i=0}^\infty e^{-\mu_a}\frac{(\mu_a)^i}{i!} e_i Y_i^\ast.
\end{aligned}
\end{equation}
Here we implicitly assume that, to her advantage, $N^a$ is known to Eve, and
\begin{equation} \label{def:yi:ei}
\begin{aligned}
Y_i^\ast &=\frac{M_i}{N_i^{\infty}},\\
e_iY_i^\ast &=\frac{e_iM_i}{N_i^{\infty}},\\
\end{aligned}
\end{equation}
where
 \begin{equation} \label{def:yi:ei11}
\begin{aligned}
N_i^{\infty}=\frac{e^{-\mu}\mu^iq^s+e^{-\nu}\nu^iq^w+q^v0^i}{i!}N
\end{aligned}
\end{equation}
is the asymptotic limit of $N_i$ when $N \rightarrow \infty$. Alternatively, we can think of $N_i^\infty$ as the expected number of $i$-photon states sent by Alice. Note that $e_iY_i^\ast$ should be regarded as one variable. Equation~\eqref{QKD-ratenew:signal} can be expanded as follows
\begin{equation} \label{VW:QYEe}
\begin{aligned}
\mathbb{E}[Q^s]&=e^{-\mu}Y_0^\ast+\mu e^{-\mu}Y_1^\ast+\frac{\mu^2e^{-\mu}}{2!}Y_2^\ast+\dots+\frac{\mu^ie^{-\mu}}{i!}Y_i^\ast+\dots\\
\mathbb{E}[Q^w]&=e^{-\nu}Y_0^\ast+\nu e^{-\nu}Y_1^\ast+\frac{\nu^2e^{-\nu}}{2!}Y_2^\ast+\dots+\frac{\nu^ie^{-\nu}}{i!}Y_i^\ast +\dots\\
\mathbb{E}[Q^v]&=Y_0^\ast\\
\mathbb{E}[E^sQ^s]&=e^{-\mu}e_0Y_0^\ast+\mu e^{-\mu}e_1Y_1^\ast+\frac{\mu^2e^{-\mu}}{2!}e_2Y_2^\ast+\dots+\frac{\mu^ie^{-\mu}}{i!}e_iY_i^\ast +\dots\\
\mathbb{E}[E^wQ^w]&=e^{-\nu}e_0Y_0^\ast+\nu e^{-\nu}e_1Y_1^\ast+\frac{\nu^2e^{-\nu}}{2!}e_2Y_2^\ast+\dots+\frac{\nu^ie^{-\nu}}{i!}e_iY_i^\ast +\dots\\
\mathbb{E}([E^vQ^v]&=e_0Y_0^\ast.\\
\end{aligned}
\end{equation}

In order to find the bounds of $M_1$ and $e_1 M_1$ in our original problem, we find the corresponding bounds for $Y_1^\ast$ and $e_1Y_1^\ast$ by calculating $\mu^2e^{\nu}\mathbb{E}[Q^w]-\nu^2e^{\mu}\mathbb{E}[Q^s]$ to obtain
\begin{equation} \label{VW:Yieibd}
\begin{aligned}
Y_1^\ast &\ge Y_1^{\ast L}=\frac{\mu}{\mu\nu-\nu^2}\left( \mathbb{E}^L[Q^w]e^\nu-\mathbb{E}^U[Q^s]e^\mu\frac{\nu^2}{\mu^2}-\frac{\mu^2-\nu^2}{\mu^2}\mathbb{E}^U[Q^v]\right),\\
e_1Y_1^\ast &\le (e_1Y_1^{\ast})^U=\frac{\mathbb{E}^U[E^wQ^w]-\mathbb{E}^L[E^vQ^v]e^{-\nu}}{\nu e^{-\nu}},
\end{aligned}
\end{equation}
which results in the following
\begin{equation} \label{result:yi:ei}
\begin{aligned}
&M_1^L=Y_1^{\ast L}N(e^{-\mu}\mu q^s+e^{-\nu}\nu q^w),\\
&(e_1M_1)^U=(e_1Y_1^\ast)^UN(e^{-\mu}\mu q^s+e^{-\nu}\nu q^w),\\
&e_1^U= \frac{(e_1M_1)^U}{M_1^L}=\frac{(e_1Y_1^\ast)^U}{Y_1^{\ast L}}=\frac{\mathbb{E}^U[E^wQ^w]e^{\nu}-\mathbb{E}^L[E^vQ^v]}{Y_1^{\ast L} \nu }.
\end{aligned}
\end{equation}

The interesting point about Eqs.~\eqref{QKD-ratenew:signal} and \eqref{VW:QYEe} is that, by some simple substitutions, they have the same form as Eq.~(13) in \cite{XMA:PRA:2005}. In fact, by replacing $\mathbb{E}[Q^a]$ ($\mathbb{E}[E^a Q^a]$) and $Y_i^\ast$ in Eq.~\eqref{QKD-ratenew:signal} with $Q_{\nu_m}$ ($E_{\nu_m} Q_{\nu_m}$) and $Y_i$, we reach to the same result as in Eq.~(13) in \cite{XMA:PRA:2005}. Note that the definitions for $Q$ and $Y$ terms here, in our finite-key analysis, are slightly different from the definitions given in \cite{XMA:PRA:2005} for the infinite-key scenario. Nevertheless, the equations look similar, and one can use the analytical results obtained in \cite{XMA:PRA:2005}, after necessary substitution, and recycle them here. For instance, the bounds obtained in Eq.~\eqref{VW:Yieibd} can directly be obtained from Eqs.~(34) and (37) in \cite{XMA:PRA:2005}.

Thus far, we have shown that the formulation that we need in either the finite-key analysis here and in \cite{curty2014finite}, or the infinitely-long key case in \cite{XMA:PRA:2005} will both result to solving a similar optimization problem. That is, once one specifies, in our formulation, the values of $\mathbb{E}^L[M^a]$, $\mathbb{E}^U[M^a]$, $\mathbb{E}^L[E^a M^a]$, and $\mathbb{E}^U[E^a M^a]$ in Eq.~\eqref{VW:MEMeqs} (or the corresponding values in other formulations), all optimization problems would result in an identical key rate estimation. The key difference would be in their estimated failure probability. The latter is a function of how we estimate the lower and upper bounds of the average terms that we need in Eq.~\eqref{VW:MEMeqs} as a function of our observations. In \cite{XMA:PRA:2005}, the authors use a heuristic Gaussian assumption, which is not exact but convenient to use. In \cite{curty2014finite}, the required bounds are obtained by using Chernoff and Hoeffding inequalities, which are rigorous but a bit too loose in certain regions. In our work, we obtain tighter bounds for these average terms, which, not only  are rigorous, but also offer higher key rates and/or lower failure probabilities as compared to the Chernoff+Hoeffding method.


\section{Statistical fluctuation analysis}\label{Sec:Statistical}
In this section, we first provide a step-by-step instruction on how to use our theoretical results in a real experimental setup. We then summarize all the tools that we have developed in our statistical fluctuation analysis. The full derivations for each of these tools will appear in Appendixes \ref{App:fromX2E} and \ref{APP:theta}.

\subsection{Instructions for experimentalists}\label{App:Parameter}
Suppose we run a QKD experiment according to the decoy-state scheme, as formulated here. After sifting and error correction, we will then have certain observables, namely, $M^{az}$ and $E^{az}$. The next step in the procedure is to apply sufficient privacy amplification that guarantees a failure probability below a given threshold $\varepsilon$. In the privacy amplification procedure, the length of the extracted secure key and hence, the size of the corresponding universal hashing function are determined by $M_1^{sz}$ and $e_1^{psz}$. Thus we need to estimate these two parameters before performing privacy amplification. Note that it is common to estimate the phase error rate $e_1^{psz}$ by using the observed bit error rate $e_1^{bsx}$ in its complement basis \cite{PhysRevLett.85.441}. One should, however, account for deviations from the bit error rate value once finite-key issues are considered \cite{Ma2011Finite}, as we do here. In this section, we only calculate the length of the secure key, $K^z$, extracted from the $Z$-basis measurements. The key length extracted from the $X$ basis, $K^x$, can be obtained similarly and the final key length is given by $K^z+K^x$. We assume that all the secure key bits come from the signal states. The final key length, $K^z$, is given by
\begin{equation} \label{VM:KeyBitZ}
\begin{aligned}
K^z &\ge M_{1}^{szL}[1-h(e_{1}^{pszU})]-K_{ec}^{sz},\\
K_{ec}^{sz}&=M^{sz}fh(E^{sz}),\\
\end{aligned}
\end{equation}
where the lower bound $M_1^{szL}$ and the upper bound $e_1^{pszU}$ can be found by taking the following steps:
\begin{enumerate}
\item Calculate ${K_{ec}^{sz}}:\\$
The parameters $M^{sz}$ and $E^{sz}$ can be directly obtained in the experiment. The cost of error correction is $K_{ec}^{sz}=M^{sz}fh(E^{sz})$.

\item Calculate ${M_1^{zL}}$ and ${e_{1}^{bxU}}$:\\
Use the results of Sec.~\ref{Sec:XtoE} to calculate the upper and lower bounds of all the average terms in Eq.~\eqref{VW:MEMeqs}, i.e., $\mathbb{E}^L[M^a]$, $\mathbb{E}^U[M^a]$, $\mathbb{E}^L[E^aM^a]$, and $\mathbb{E}^U[E^aM^a]$ for each basis. Then use $\mathbb{E}[Q^{a}]=\mathbb{E}[M^{a}]/N^{a}$ and $\mathbb{E}[E^aQ^{a}]$=$\mathbb{E}[E^aM^{a}]/N^{a}$ to calculate the corresponding $Q$ and $EQ$ parameters.
Then, use Eqs.~\eqref{VW:Yieibd} and \eqref{result:yi:ei} to calculate ${M_1^{zL}}$ and ${e_{1}^{bxU}}$.

\item Calculate ${M_{1}^{szL}}:\\$
Use Eq.~\eqref{simplier standard lower bound} in Sec.~\ref{Section:EtoX} to calculate $M_{1}^{szL} = \chi^L$ for $\bar \chi = p_1^s M_1^{zL}$.


\item Calculate ${e_{1}^{pszU}}:\\$
Use Eq.~\eqref{upper bound Random sampling} to find ${e_{1}^{pszU}}$. In Appendix \ref{APP:theta}, we use the random sampling method to account for the deviation, $\theta$, between $e_{1}^{bx}$ and $e_{1}^{psz}$ caused by the finite-key setting in our problem. The upper bound on $e_{1}^{bx}$ has already obtained in Step 2. By upper bounding $\theta$ as explained in Appendix \ref{APP:theta}, we can find ${e_{1}^{pszU}}$. This will specify the required amount of privacy amplification in the protocol.

\end{enumerate}



\subsection{Methodology: Key ideas}\label{Sec:chernoff}

The first nontrivial step in our instruction list, given in Sec.~\ref{App:Parameter}, is to calculate lower and upper bounds for all the average terms of interest. The key idea to solve this problem, in our case, is to use the Chernoff bound with an inverse formulation. To make this point clear, in this section, we first review the Chernoff bound in the special case of Bernoulli random variables and show that why it is relevant to our problem. Then, by rewriting the Chernoff bound, we find proper candidates for upper and lower bounds of the relevant average terms. In the end, we comment on the differences between our approach and that of \cite{curty2014finite}.

The Chernoff bound for a set of $n$ independent Bernoulli random variables $\chi_i\in\{0,1\}$ can be expressed as follows \cite{bound1,bound2}. For $\chi=\sum_{i=1}^n \chi_i$ and $\bar\chi=\mathbb{E}[\chi]$, we have the following bounds
\begin{equation} \label{Chernoff bound}
\begin{aligned}
\Pr[\chi >(1+\delta^L) \bar\chi ]<\left[\frac{e^{\delta^L}}{(1+\delta^L)^{1+\delta^L}}\right]^{\bar\chi} = g(\delta^L,\bar\chi),
\end{aligned}
\end{equation}
and
\begin{equation} \label{App:EXup}
\begin{aligned}
\Pr[\chi<(1-\delta^U)\bar\chi]<\left[\frac{e^{-\delta^U}}{(1-\delta^U)^{1-\delta^U}}\right]^{\bar\chi} = g(-\delta^U,\bar\chi),
\end{aligned}
\end{equation}
where $\delta^L>0$, $0<\delta^U <1$, and $g(\delta, \bar\chi) = \left[\frac{e^{\delta}}{(1+\delta)^{1+\delta}}\right]^{\bar\chi}$.

The above formulation can be applied to $M^a$ and $E^a M^a$, whose average values need to be bounded. For instance, in the data postprocessing step, the total number of detections obtained by Bob in the $Z$ basis is given by $M^z$. For each valid detection event, we can define the indicator random variable $\chi_j$ that determines whether or not Alice has originally prepared the $j$th received pulse in the {\em signal} state. That is, $\chi_j=1$ means that a signal state has caused the $j$th detection event, whereas $\chi_j=0$ implies that another state (weak decoy or vacuum state) has been used. Then, the total number of detected {\em signal} states is given by $M^{sz}=\sum_{j=1}^{M^z} \chi_j$, with $\chi_j$ being independent Bernoulli random variables. A similar formulation can be used for error terms as well. In the rest of this section, the parameter $\chi$ will then represent any of the parameters of interest in the form $M^a$ and $E^a M^a$ in a particular basis.

The Chernoff bound in Eqs.~\eqref{Chernoff bound} and \eqref{App:EXup} bounds the probability that the observed value deviates from its average value. That is, if we know the average value of $\chi$, we can define a confidence interval $[\chi^U,\chi^L]$, where $\chi^L = (1+\delta^L) \bar \chi$ and $\chi^U = (1-\delta^U) \bar \chi$, the probability of being outside of which is bounded by functions of $\delta^L$, $\delta^U$, and $\bar\chi$. The problem that we have in hand is, however, the opposite. We need to bound $\bar\chi$ for a given observed value of $\chi$ in such a way that the failure probability is below a certain threshold.

To define the failure probability precisely, we use the same framework that we developed in Sec.~\ref{Sub:VMOpt} in which we showed that after the measurement phase, $\bar \chi$ is fixed, but unknown. Nevertheless, even for a fixed $\bar \chi$, the value $\chi$ that Alice and Bob observe in their experiment is a random variable. The failure probability in this setting can then be defined as follows. For a fixed but unknown value of $\bar \chi$, we find the probability that the observed value for $\chi$ results in either of the following events:
\begin{equation}
\mbox{Event 1: } \bar \chi < \mathbb{E}^L(\chi),
\end{equation}
where $\mathbb{E}^L(\chi)$ is the procedure/function by which we relate an observed value to the lower limit on $\bar \chi$, and
\begin{equation}
\mbox{Event 2: } \bar \chi > \mathbb{E}^U(\chi),
\end{equation}
where $\mathbb{E}^U(\chi)$ is the procedure/function by which we relate an observed value to the upper limit on $\bar \chi$. For instance, the probability of failure corresponding to Event 1 is given by
\begin{equation}
\Pr[\mbox{Event 1}] = \Pr[\bar \chi < \mathbb{E}^L(\chi)].
\end{equation}
Now, in order to bound the above probability, we define our function $\mathbb{E}^L(\chi)$ in such a way that it satisfies the following condition
\begin{equation}
\label{PrEL}
\Pr[\chi >(1+\delta^L(\varepsilon^L, \bar\chi)) \bar\chi ] = \Pr[\bar\chi < \mathbb{E}^L(\chi)] ,
\end{equation}
where $\varepsilon^L$, as we see next, is the failure probability, and we have solved the equation $g(\delta^L,\bar\chi) = \varepsilon^L$ in order to write $\delta^L$ as a function of $\varepsilon^L$ and $\bar\chi$. The left-hand-side of Eq.~\eqref{PrEL} is then equivalent to the left-hand-side of Eq.~\eqref{Chernoff bound}, which will then result in
\begin{equation}
\Pr[\mbox{Event 1}] < \varepsilon^L.
\end{equation}
In other words, by choosing $\mathbb{E}^L(\chi)$ in such a way that it satisfies Eq.~\eqref{PrEL} we can use the Chernoff bound to bound the failure probability. The same holds if one works out the upper limit for the average terms with the difference that now one should find $\mathbb{E}^U(\chi)$ such that
\begin{equation}
\label{PrEU}
\Pr[\chi >(1-\delta^U(\varepsilon^U, \bar\chi)) \bar\chi ] = \Pr[\bar\chi < \mathbb{E}^U(\chi)] ,
\end{equation}
with $\varepsilon^U$ being the failure probability for Event 2 and $\delta^U(\varepsilon^U, \bar\chi)$ is the solution to $g(-\delta^U,\bar\chi) = \varepsilon^U$.

Provided that functions $\chi^L = (1+\delta^L(\varepsilon^L, \bar\chi)) \bar \chi$ and $\chi^U = (1-\delta^U(\varepsilon^U, \bar\chi) \bar \chi$ are increasing functions of $\bar \chi$, one obvious choice for $\mathbb{E}^L(\chi)$ ($\mathbb{E}^U(\chi)$) is the inverse function of $\chi^L$ ($\chi^U$). In Appendix~\ref{App:fromX2E}, we show that the above monotonicity condition, in fact, holds, and that would offer a solution to find very tight bounds for all terms of interest.

Our approach offers tighter bounds than the ones proposed in \cite{curty2014finite}. One reason for the difference is that, in \cite{curty2014finite}, the authors use looser forms of the Chernoff bound than the ones we use in Eqs.~\eqref{Chernoff bound} and \eqref{App:EXup}, especially when $\chi$ has small values. But, more importantly, the procedure for finding $\mathbb{E}^U(\chi)$ in \cite{curty2014finite} is somehow heuristic, as compared to our exact calculations, and results in looser upper bounds even in the case of large values of $\chi$. In our numerical results we show how these differences will result in our improving the bounds, and correspondingly the failure rate and/or key rate, in the decoy-state QKD setup. In the rest of this section, we then provide a summary of our analytical results that can be used to bound relevant terms in our formulation.



\subsection{From $\chi$ to $\bar \chi$}\label{Sec:XtoE}
Given a measurement result $\chi$, we can bound the underlying expectation value $\bar \chi$ for a failure probability bounded by $\varepsilon = 2 \varepsilon^L = 2 \varepsilon^U$. The results are summarized below and the details of calculations are shown in Appendix \ref{App:fromX2E}.


\begin{enumerate}
\item
If $\chi=0$, we use
\begin{equation} \label{confidence interval11}
\begin{aligned}
\mathbb{E}^L(\chi)&=0, \\
\mathbb{E}^U(\chi)&=\beta,
\end{aligned}
\end{equation}
where $\beta = -\ln(\varepsilon/2)$.
\item
If $\chi> 0$, we use
\begin{equation} \label{Chernoff:Xg0Ebds}
\begin{aligned}
\mathbb{E}^L(\chi)&=\frac{\chi}{1+\delta^L},\\
\mathbb{E}^U(\chi)&=\frac{\chi}{1-\delta^U}, \\
\end{aligned}
\end{equation}
where $\delta^L$ and $\delta^U$ can be obtained by solving the following equations
\begin{equation} \label{Chernoff:Xg0deltas}
\begin{aligned}
&\left[\frac{e^{\delta^L}}{(1+\delta^L)^{1+\delta^L}}\right]^{\frac{\chi}{1+\delta^L}}=\frac12\varepsilon,\\
&\left[\frac{e^{-\delta^U}}{(1-\delta^U)^{1-\delta^U}}\right]^{\frac{\chi}{1-\delta^U}}=\frac12\varepsilon.
\end{aligned}
\end{equation}
It turns out that the solutions $\delta^L$ and $\delta^U$ to Eq.~\eqref{Chernoff:Xg0deltas} are difficult to calculate when $\chi$ is large. A simplified analytical approximation is given next.
\item
If $\chi\ge 6 \beta$, we use
\begin{equation} \label{confidence interval31}
\begin{aligned}
\delta^L=\delta^U&=\frac{3\beta+\sqrt{8\beta\chi+\beta^2}}{2(\chi-\beta)} \\
\end{aligned}
\end{equation}
in Eq.~\eqref{Chernoff:Xg0Ebds}. This will provide us with a slightly looser bound than the one we can obtain by solving \eqref{Chernoff:Xg0deltas}, but the difference is negligible.

\end{enumerate}

\subsection{From $\bar \chi$ to $\chi$}\label{Section:EtoX}

Once, using the relationships in Sec.~\ref{Sec:XtoE}, $\mathbb{E}^L(\chi)$ and $\mathbb{E}^U(\chi)$ are found for all relevant parameters $\chi$, we use Eqs.~\eqref{VW:Yieibd} and \eqref{result:yi:ei} to calculate ${M_1^{zL}}$ and ${e_{1}^{bxU}}$. In step 3 of the instruction list, we, however, need to calculate $M_1^{szL}$. We know that $\mathbb{E}[M_1^{sz}]=p_1^{sz} M_1^z$. In this section, we will show, using a symmetric form of the Chernoff bound, how to estimate the value of $M_1^{sz}$ from $\mathbb{E}[M_1^{sz}]$.

Let us use our more general notation $\chi$ representing the sum of a number of independent Bernoulli random variables. $M_1^{sz}$ satisfies this condition as written in Eq.~\eqref{M_ia}. Then, we can solve the following equation
\begin{equation} \label{simplier standard Chernoff bound root}
\begin{aligned}
2e^{-\delta^2\bar \chi/(2+\delta)}=\varepsilon, \\
\end{aligned}
\end{equation}
and, using the symmteric form of the Chernoff bound given by \cite{Lecture1,Lecture2}
\begin{equation} \label{simplier Chernoff bound}
\Pr(|\chi- \bar \chi|\ge \delta \bar \chi) \le 2e^{-\delta^2\bar \chi/(2+\delta)},
\end{equation}
we obtain a confidence interval $[\chi^L,\chi^U]$, for which $\Pr \{\chi \in [\chi^L,\chi^U] \}> 1- \varepsilon$, where
\begin{equation} \label{simplier standard lower bound}
\begin{aligned}
\chi^L &=(1-\delta)\bar \chi, \\
\chi^U &=(1+\delta)\bar \chi, \\
\delta &=\frac{-ln(\varepsilon/2)+\sqrt{(\ln(\varepsilon/2))^2-8\ln(\varepsilon/2)\bar \chi}}{2\bar \chi}. \\
\end{aligned}
\end{equation}

In our problem, we have the lower bound for $\bar \chi = \mathbb{E}[M_1^{sz}]$ given by $p_1^{sz} M_1^{zL}$. We can then use the relationship for $\chi^L$ above to calculate $M_1^{szL}$ with a failure probability bounded by $\varepsilon$.





\section{Numerical results}\label{Sec:Simulation}

In this section, we provide additional insight into our proposed method by numerically comparing it with the other two methods of Chernoff+Hoefding and the Gaussian analysis. We compare the three methods in terms of the tightness of their confidence intervals, or their failure probability, as well as the secret key generation rate and the maximum secure distance in the finite-key setting.

\subsection{Tightness of the bounds}
Here, we compare the two previously proposed methods in \cite{XMA:PRA:2005} and \cite{curty2014finite} with ours in terms of bounding the expectation value $\mathbb{E}[\chi]$, from an observation value $\chi$. For ease of reference, we have summarized the Gaussian analysis method in Appendix \ref{App:StdErrAna} and the Chernoff+Hoeffding method \cite{curty2014finite} in Appendix \ref{App:Curty}. For different methods, we calculate the width of the confidence interval for a fixed failure probability $\varepsilon$. We define this width as $d=(\mathbb{E}^U[\chi]-\mathbb{E}^L[\chi])/2$, which quantifies the tightness of an analysis method. Below, we consider the two extreme cases of large and small value of $\chi$.  



Figure~\ref{Fig:Deviation1} compares the three methods in terms of the width of the confidence interval $d$ for different failure probabilities when the observed value is rather large. We have normalized the vertical axis by $\sigma=\sqrt{\chi}$, which, for $\chi\rightarrow\infty$, is somehow a measure of standard deviation for the original random variable. Among the three methods, the Gaussian analysis method gives the tightest bounds, but that comes at the price of not being able to rigorously bound the failure rate. Our proposed method almost follows that of the Gaussian curve, while there is a considerable gap between our method and the Chernoff-Hoeffding one. This implies that the latter offers looser bounds on the average terms of interest as compared to our proposed technique.




\begin{figure}[hbt]
\centering \resizebox{12cm}{!}{\includegraphics{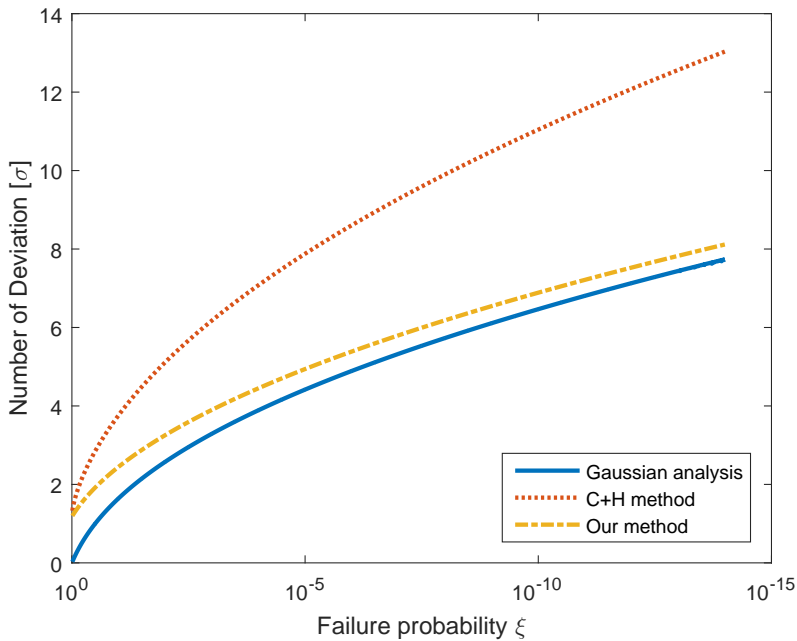}}
\caption{Comparison of the width of the confidence interval versus failure probability for three methods: the Gaussian analysis (solid), the Chernoff+Hoeffding \cite{curty2014finite} (dotted), and our new method (dash-dotted). In each scheme, we find lower and upper bounds for the expectation value $\mathbb{E}[\chi]$ from an observed value $\chi$, at a given failure probability and at $\chi \rightarrow\infty$. The vertical axis then represents $(\mathbb{E}^U[\chi]-\mathbb{E}^L[\chi])/(2\sigma)$, for $\sigma = \sqrt{\chi}$.}
\label{Fig:Deviation1}
\end{figure}

We also compare the three fluctuation analysis methods from another perspective where we fix the fluctuation deviations, $\chi-\mathbb{E}^L[\chi]$ or $\mathbb{E}^U[\chi]-\chi$, and evaluate the failure probabilities. 
The results are shown in Table~\ref{tab:deviations}. We find that in the Chernoff+Hoeffding method \cite{curty2014finite}, the failure probability for Event 2, at an identical deviation, is higher than that of Event 1. This is because, in their formulation, $\chi-\mathbb{E}^L[\chi]\ne\mathbb{E}^U[\chi]-\chi$, and their estimate of the upper bound, $\mathbb{E}^U[\chi]$, is rather loose. For large values of $\chi$, the failure probability for both events is the same for our method as well as the Gaussian analysis one.
It can be seen that the failure probability guaranteed by our method is roughly within one order of magnitude of that of the Gaussian analysis method. Note that, however, in the latter case, the failure probabilities are not guaranteed and they rely on an underlying Gaussian assumption, which is not necessarily the case. Table~\ref{tab:deviations} can then serve as a guideline from which one can specify the desired failure probability and then quickly estimate the corresponding values for $\mathbb{E}^L[\chi]$ and $\mathbb{E}^U[\chi]$.
\begin{table} [htb]
\centering
\caption{The failure probability as a function of the fluctuation deviations, $\chi-\mathbb{E}^L[\chi] = \mathbb{E}^U[\chi]-\chi$ when $\chi\rightarrow\infty$. Here, $\varepsilon_{G}$,  $\varepsilon_{C+H}$, and $\varepsilon_{new}$, respectively, denote the sum failure probability for Events 1 and 2 for the Gaussian analysis, the Chernoff+Hoeffding method \cite{curty2014finite}, and our new method.}\label{tab:deviations}
\begin{tabular}{c|cccc}
\hline
\hline
Deviation & $\varepsilon_{G}$ & $\varepsilon_{C+H}$ & $\varepsilon_{new}$  \\ \hline
$3\sigma$ & $10^{-2.56}$ & $10^{-0.57}$ & $10^{-1.65}$ \\
$5\sigma$ &$10^{-6.24}$ &$10^{-1.90} $&$10^{-5.12}$\\
$7\sigma$ &$10^{-11.59} $&$10^{-3.90}$ &$10^{-10.33}$\\
$9\sigma$ &$10^{-18.64} $ &$10^{-6.57}$&$10^{-17.28}$\\
\hline
\hline
\end{tabular}
\end{table}

Our method is particularly attractive when the observed counts are small. As shown in Figure \ref{Fig:smallx}, we compare our method with the Gaussian analysis, at a fixed failure probability of $\varepsilon=10^{-10}$, in terms of lower and upper bounds on the expectation value $\mathbb{E}[\chi]$ when the observed value for $\chi$ is small. When estimating the upper bound, the Gaussian analysis is always tighter than our new method. When $\chi\rightarrow 0$, the upper bound of the Gaussian analysis is $0$ and that of our new method is $23.7190$, which is equal to the value of $\beta$ at $\varepsilon=10^{-10}$. Our method, nevertheless, offers a tighter estimation of the lower bound for $\chi < 2257$. In comparison with the Chernoff+Hoeffding method, our method offers a substantial advantage in the sense that our required deviations are optimized by solving Eq.~\ref{Chernoff:Xg0deltas}, whereas in the Chernoff+Hoeffding method the deviations are proportional to the number of counts; see, e.g., Eq.~\eqref{Hoeffding inequality} in Appendix \ref{App:Curty}.


\begin{figure}[bht]
\centering \resizebox{12cm}{!}{\includegraphics{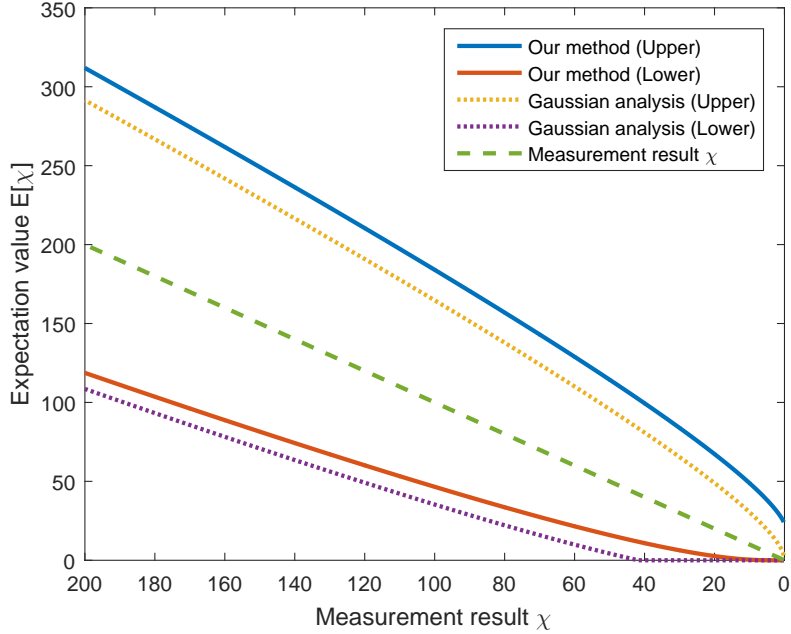}}
\caption{Lower and upper bounds of the expectation value versus observed values of $\chi$ for the Gaussian analysis (dotted) and our new method (solid). In both cases, the failure probability is fixed at $\varepsilon=10^{-10}$.}
\label{Fig:smallx}
\end{figure}


Another interesting feature of our methodology is the dependence of the failure probability on the observed value $\chi$. As shown in Fig.~\ref{Fig:Deviation1} and Table~\ref{tab:deviations}, given a fixed failure probability $\varepsilon$, the fluctuation deviation can be written as a constant multiplied by $\sigma=\sqrt{\chi}$. One could ask the opposite question that for a given fluctuation deviation of $n_\alpha \sigma$, for a fixed value of $n_\alpha$, how the failure probability would vary with $\chi$. This question has been answered in Corollary~\ref{corollary:fixed intercal:probability} and the results have been shown in Fig.~\eqref{Fig:epsfinitechi} for several different values of $n_\alpha$. It can be seen that for large values of $\chi$, the fluctuation probability approaches the constant value given in Table~\ref{tab:deviations}. For small values of $\chi$, however, the failure probability goes up as now, for the given confidence interval, the chance of making an error is higher. This is in contrast with what the Gaussian analysis method assumes in that the failure probability for a fixed value of $n_\alpha$ is independent of $\chi$; see Eq.~\eqref{deviation} in Appendix~\ref{App:StdErrAna}. This is how our method offers a more rigorous approach to the finite-key analysis as compared to the Gaussian analysis method.

\begin{figure}[hbt]
\centering \resizebox{12cm}{!}{\includegraphics{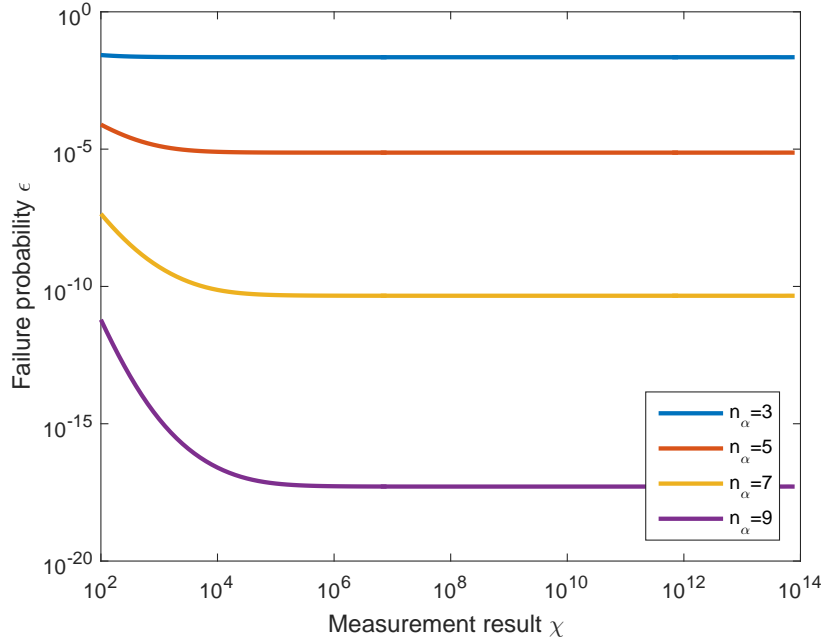}}
\caption{The total failure probability $\varepsilon$ versus the observed value $\chi$ when we fix the deviation from the mean value is given by $n_\alpha \sigma$, for $n_\alpha=3,5,7,9$ from top to bottom curves.}
\label{Fig:epsfinitechi}
\end{figure}

\subsection{Key rate comparison}
In order to compare the performance of our technique, in terms of the final key rate and the maximum secure transmission distance, with previous work, we simulate our QKD system by assuming that the observed values for different parameters of interest is given by their asymptotic values in an Eve-free experiment. These values have been summarized below \cite{XMA:PRA:2005}:
\begin{equation}\label{Yi:Qi}
\begin{aligned}
Q^{a}&=Y_0+(1-Y_0)(1-e^{-\eta\mu_a}), \\
E^{a}Q^{a}&=e_0Y_0+e_d(Q^{a}-Y_0),\\
\end{aligned}
\end{equation}
where $\eta$ is the total transmittance, $Q^{a}$ and $E^{a}$ are the overall gain and QBER, $e_d$ is the misalignment error rate, and the error rate of the background noise, $e_0$, is equal to $1/2$. Note that the values used in Eq.~\eqref{Yi:Qi} is for simulation purpose only. In a real experiment, all the variables on the left hand side can directly be measured. For the simulation of the asymptotic case with an infinite number of decoy states, where all the channel properties can be estimated accurately, we use the following formula
\begin{equation}\label{Sim:Yiei}
\begin{aligned}
Y_i &=1-(1-Y_0)(1-\eta)^i, \\
e_iY_i &=e_0Y_0+e_d(Y_i-Y_0),
\end{aligned}
\end{equation}
where $Y_i$ and $e_i$  are the yield and the error rate of the $i$-photon channel.

In our numerical results, we optimize the choice of the intensities and the ratios of the signal, weak decoy, and vacuum state to maximize the final key rate. To perform parameter optimization, the local search algorithm (LSA) \cite{boyd2004convex} is employed. In the following simulation, we use the parameters of a practical QKD system \cite{GYS}, as listed in Table \ref{tab:parameters}. Note that, in our work, $\varepsilon$ represents the failure probability of each step. In our method, the failure probability of a single upper (lower) bound is $\varepsilon/2$ and therefore, the failure probability of a confidence interval, composed of an upper bound and a lower bound, is $\varepsilon$. The total failure probability of the whole QKD system (including both $X$ and $Z$ bases) is $8\varepsilon$.

\begin{table}[hbpt]
\centering  %
\caption{Parameters for a practical QKD system where $\eta_d$ is the detection efficiency, $f$ is the inefficiency of error correction, and $N$ is the number of pulses sent by Alice.
}\label{tab:parameters}
\begin{tabular}{lcccccc}
\hline
$\eta_d$  &$Y_0$ & $f$& $e_d$& Loss &$\varepsilon$ & $N$ \\
\hline
4.5\%  &$1.7\times10^{-6}$ &1.22 ~&$3.3\%$&0.21 dB/km &$10^{-10}$ & $10^{10}$ \\
\hline
\end{tabular}
\end{table}

We compare the three discussed fluctuation analysis methods with the asymptotic case, where, in the latter, the data size is infinitely large and its statistical fluctuations can be ignored. The results are shown in Fig.~\ref{Fig:RdisComp}. It is clear that our new method always provides a larger final key rate than the Chernoff+Hoeffding method \cite{curty2014finite}. For $N=10^{10}$, our analysis method increases the maximum secure transmission distance by $7$ km. In the limit of short transmission distances, the number of pulses detected by Bob is very large, and, therefore, the improvement of our new method is not substantial. In the regime around the maximum secure transmission distance, the value of $\chi$ is small and our new method is advantageous. Meanwhile, from Fig.~\ref{Fig:RdisComp}, one can clearly see that our new method achieves a very close performance to the widely-used Gaussian analysis method \cite{XMA:PRA:2005}.

\begin{figure}[tbh]
\centering \resizebox{12cm}{!}{\includegraphics{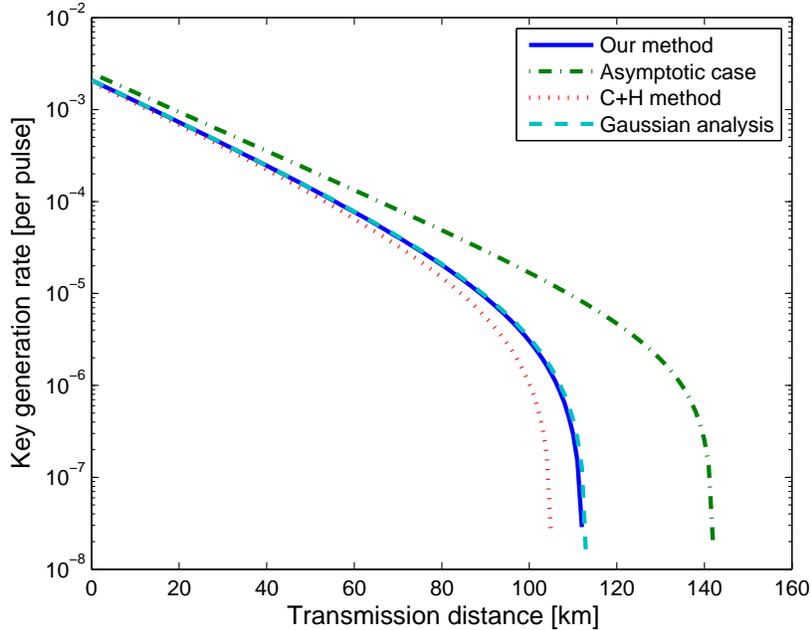}}
\caption{Comparison of the key rates obtained by the three methods, the Gaussian analysis, the Chernoff+Hoeffding method \cite{curty2014finite}, and our new method. The infinite key length case is also shown in this figure.}
\label{Fig:RdisComp}
\end{figure}

For our method, at short QKD distances, the optimized intensity of the signal state $\mu$ is equal to $0.45$. As the distance increases, the optimum intensity of the signal state decreases. At a distance of $100$ km, the optimized $\mu$ decreases to $0.37$ with other optimized parameters listed in Table \ref{Tab:parameters}. All the results are consistent with the Gaussian analysis case \cite{XMA:PRA:2005}. 

\begin{table}[hbpt]
\centering  %
\caption{Optimized parameters at 100 km.} \label{Tab:parameters}
\begin{tabular}{ccccccc}
\hline
Key rate &$\nu$  &$\mu$ & $p_{\nu}$& $p_{\mu}$ \\
\hline
$3.04\times 10^{-6}$ &$0.126$&$0.370$&$0.250$&$0.650$ \\
\hline
\end{tabular}
\end{table}

Finally, in Fig.~\ref{Fig:distance}, we consider the relation between the data size and the corresponding maximum secure transmission distance for all three methods disucssed. When the total data size of a QKD protocol is larger than $10^{14}$, its maximum secure transmission distance is very close to the asymptotic limit of 142 km. No secret keys can be exchanged at a data size, $N$, roughly below $10^{7}$. The curves of our method and the Gaussian case are almost the same. When N is smaller than $10^{12}$, all three curves are very steep. Consequently, the gap between maximum secure transmission distances of our method and the Chernoff+Hoeffding method is distinct. For example, as shown in both Fig.~\ref{Fig:RdisComp} and Fig.~\ref{Fig:distance}, our method increases the maximum transmission distance by $7$ km when total data size $N=10^{10}$.

\begin{figure}[tbh]
\centering \resizebox{12cm}{!}{\includegraphics{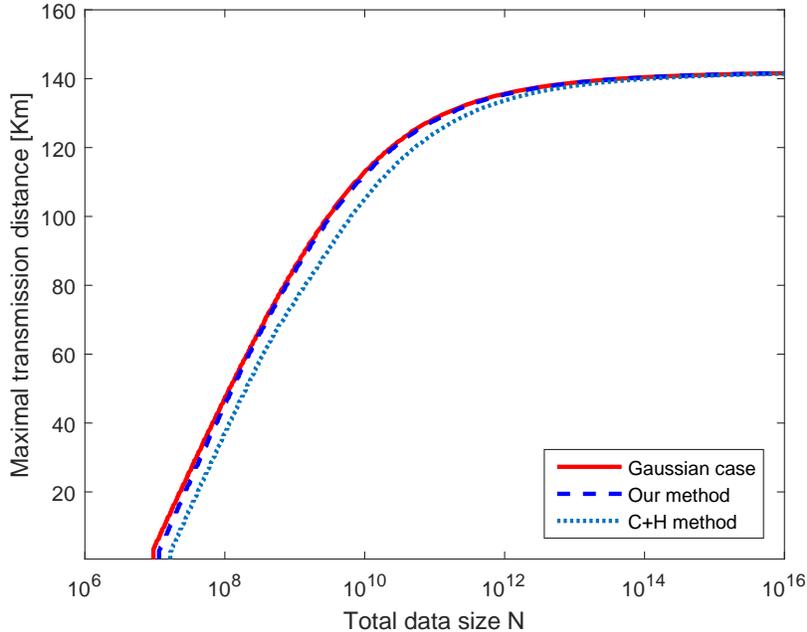}}
\caption{Maximum secure transmission distance versus the number of pluses sent by Alice, $N$. The simulation parameters are listed in Table \ref{Tab:parameters}. No secure keys can be generated for $N\le10^{7}$. The asymptotic limit for the maximum secure transmission distance is $142$ km when $N\ge10^{14}$.}
\label{Fig:distance}
\end{figure}


\section{Conclusions and Discussion}\label{Sec:Conclusion}
In this paper, we developed a tight bound for the decoy-state QKD system when the finite-data-size effects are taken into account. As compared to the early work on this topic, which relied on Gaussian approximations, our method offered a rigorous approach to estimating the failure probability. In that sense, our method was similar to the recently proposed techniques relying on Chernoff and Hoeffding inequalities. Our proposed method could, however, substantially improve the performance by yielding a smaller failure probability, for a similar confidence interval, than what the Chernoff+Hoeffding method could offer. In fact, after parameter optimization, our method could offer similar performance to the widely-used Gaussian analysis method, which uses non-rigorous Gaussian approximations.


There are several problems to which our methodology can be applied. In this work, we assumed that the phase of the weak coherent state was continuously randomized. When the phase is not randomized, we know that security loopholes may allow for certain attacks \cite{LoPreskill:NonRan:2007,Tang2013Source}. In practice, it is difficult to randomize the phase of a laser pulse continuously. Instead one can apply the discrete phase randomization \cite{Discretephase2014}, using which the final secure key rate is slightly reduced. Our finite-key analysis for the decoy-state method can then be applied to the discrete phase randomization case. Our method is also applicable to the biased BB84 protocol \cite{Wei2013DecoyBiased}, in which the choice of basis is not symmetric. The analysis method in this work can also be used in other protocols, such as measurement-device-independent QKD protocol \cite{Lo:MIQKD:2012,Braunstein2012MDIQKD} and round-robin differential-phase-shift QKD protocol \cite{sasaki2014practical,zhang2015round}. We expect that our methodology will offer similar performance to the Gaussian analysis method, while the security parameters have been rigorously estimated. In addition to finite-size effects, laser source intensity fluctuations should also be taken into consideration in practice \cite{wang2008,wang2009}. It is important to investigate all these practical issues together for QKD systems.

\section{Acknowledgments}
The author acknowledges insightful discussions with Z.~Cao, M.~Curty, C.-H.~F.~Fung, H.-K.~Lo, N.~L\"utkenhaus, and X.~Yuan. This work was supported by the 1000 Youth Fellowship program in China and the UK EPSRC Grant No. EP/M013472/1.

\appendix

\section{From $\chi$ to $\bar \chi$}\label{App:fromX2E}

\subsection{Chernoff bound method}\label{APP:Chernoff bound method}
In this section, we provide a confidence interval for the expectation value ${\bar \chi}$ based on the observed value $\chi$. We use the methodology described in Sec.~\ref{Sec:chernoff} and the original forms of the Chernoff bound in Eqs.~\eqref{Chernoff bound} and \eqref{App:EXup}. Our proposed method works even if $\chi$ approaches $0$, and unlike the Chernoff+Hoesffding method, we do not need to use the Hoeffding inequality in this regime. 
Without loss of generality, we assume that the failure probabilities for Events 1 and 2 are equal and are given by $\varepsilon/2$. the total failure probability in bounding the expected values is then given by $\varepsilon$. As mentioned in Sec.~\ref{Sec:chernoff}, the lower and upper bounds on ${\bar \chi}$ can be obtained by, respectively, solving the following set of equations:
\begin{equation} \label{successful4}
\begin{aligned}
g(\delta^L,\bar\chi)&=[\frac{e^{\delta^L}}{(1+\delta^L)^{1+\delta^L}}]^{\bar \chi}=\varepsilon/2,\\
{\bar \chi}&=\frac{\chi}{1+\delta^L}, \\
\delta^L&\ge0,\\
\end{aligned}
\end{equation}
and 
\begin{equation} \label{successful5}
\begin{aligned}
g(-\delta^U,\bar\chi)&=[\frac{e^{-\delta^U}}{(1-\delta^U)^{1-\delta^U}}]^{\bar \chi}=\varepsilon/2,\\
{\bar \chi}&=\frac{\chi}{1-\delta^U}, \\
0&<\delta^U<1, \\
\end{aligned}
\end{equation}
or equivalently, for given values of $\chi$ and $\varepsilon$, we need to solve the following two equations
\begin{equation} \label{shortereqns}
\begin{aligned}
g(\delta^L, \chi/(1+\delta^L)) &= \varepsilon/2 \\
g(-\delta^U, \chi/(1-\delta^U)) &= \varepsilon/2
\end{aligned}
\end{equation}
to obtain $\delta^L$ and $\delta^U$. The lower and upper bounds of $\mathbb{E}[\chi]$ are then given by
\begin{equation} \label{confidence interval2}
\begin{aligned}
&{\mathbb{E}^L}[\chi]=\frac{\chi}{1+\delta^L},\\
&{\mathbb{E}^U}[\chi]=\frac{\chi}{1-\delta^U}.\\
\end{aligned}
\end{equation}

\begin{Claim} \label{Lemma:chernof:lower bound}
For all $\chi > 0$, there exist unique answers for $\delta^L>0$ and $0<\delta^U<1$ in Eq.~\eqref{shortereqns}.
\end{Claim}
\begin{proof}
Let us first rewrite Eq.~\eqref{shortereqns} as follows:
\begin{equation} \label{shortereqns2}
\begin{aligned}
g_2(\delta^L) &= \ln(1+\delta^L)-\delta^L/(1+\delta^L) = \beta/\chi, \\
g_2(-\delta^U) &= \ln(1-\delta^U)+\delta^U/(1-\delta^U) = \beta/\chi,
\end{aligned}
\end{equation}
where $\beta=-\ln(\varepsilon/2)\ge 0$. It is easy to verify that $g_2(0) = 0$, $g_2(\infty) = \infty$, and $g_2(-1) = \infty$. This would guarantee that there exists solutions for $\delta^L$ and $\delta^U$ in their respective regions. Furthermore, it can be verified that $g_2(\delta)$ is a monotonic function of $\delta$ in both regions of $-1<\delta<0$ and $\delta>0$. This guarantees that the solutions found are unique. This would imply that the corresponding lower and upper bounds in Eq.~\eqref{confidence interval2} would provide us with the tightest bound possible in Eqs.~\eqref{PrEL} and \eqref{PrEU}.
\end{proof}


\begin{corollary} \label{corollary:fixed intercal:probability}
For a given observed value $\chi$ and a confidence interval $[{\mathbb{E}^L}[\chi],{\mathbb{E}^U}[\chi]]$, the failure probability is given by
\begin{equation} \label{failure probability totoal}
\begin{aligned}
\varepsilon=e^{-\chi g_2(\delta^L)}+e^{-\chi g_2(-\delta^U)},
\end{aligned}
\end{equation}
where $\delta^L$ and $\delta^U$ can be obtained from Eq.~\eqref{confidence interval2}.
\end{corollary}
\begin{proof}
From Eq.~\eqref{shortereqns2}, the values of $\beta^L$ ($\beta^U$) can be calculated as follows
\begin{equation} \label{fixed intercal:beta}
\begin{aligned}
\beta^L=\chi g_2(\delta^L),\\
\beta^U=\chi g_2(-\delta^U).\\
\end{aligned}
\end{equation}
From their definition, we also have $\beta^L=-\ln(\varepsilon^L)$ and $\beta^U=-\ln(\varepsilon^U)$, where $\varepsilon^L$ ($\varepsilon^U$) is the corresponding failure probability to Event 1 (2), which results in
\begin{equation} \label{fixed intercal:varepsilon}
\begin{aligned}
\varepsilon^L=e^{-\chi g_2(\delta^L)},\\
\varepsilon^U=e^{-\chi g_2(-\delta^U)}.\\
\end{aligned}
\end{equation}
The failure probability of the given confidence interval, $\varepsilon$, is then given by $\varepsilon^L+\varepsilon^U=e^{-\chi g_2(\delta^L)}+e^{-\chi g_2(-\delta^U)}$.
\end{proof}

\begin{Claim} \label{corollary:chernof:bound}
In the limit of $\chi\rightarrow\infty$, the lower and upper bounds of $\bar \chi$ in Eq.~\eqref{confidence interval2} are given by,
\begin{equation} \label{confidence interval1new}
\begin{aligned}
{\mathbb{E}}^L[\chi]&=\chi(1 -\sqrt{\frac{2\beta}{\chi}}), \\
\mathbb{E}^U[\chi]&=\chi(1 +\sqrt{\frac{2\beta}{\chi}}).\\
\end{aligned}
\end{equation}
\end{Claim}
\begin{proof}
For large values of $\chi$, $\beta/\chi$ is small, and therefore the corresponding solutions for $\delta^L$ and $\delta^U$ would be small too. In this regime, one can use the Taylor series for the log function, up to two terms, to simplify Eq.~\eqref{shortereqns2}
to obtain
\begin{equation} \label{successful analytical solutions}
\begin{aligned}
\delta^L=\delta^U=\sqrt{\frac{2\beta}{\chi}}.
\end{aligned}
\end{equation}
The conclusion will follow if we replace the above answer into Eq.~\eqref{confidence interval2}.
\end{proof}

\subsection{Simplified result when $\chi$ is large} \label{APP:Simplified case for the Chernoff bound estimation}
In Appendix~\ref{APP:Chernoff bound method}, we showed how to tightly bound the expectation value ${\bar \chi}$. The above numerical method can, however, become tedious when $\chi$ is very large. To overcome this problem, we use the symmetric form of the Chernoff bound in Eq.~\eqref{simplier Chernoff bound} and give an explicit result in the specific case of $\chi > 6\beta$.

\begin{Claim} \label{Lemma:Simplified chernof}
For $\chi>6\beta$, the lower and upper bounds of $\bar \chi$ are given by
\begin{equation} \label{confidence interval1}
\begin{aligned}
&{\mathbb{E}}^L[\chi]=\frac{\chi}{1+\delta},\\
&\mathbb{E}^U[\chi]=\frac{\chi}{1-\delta},\\
&\delta=\frac{3\beta+\sqrt{8\beta\chi+\beta^2}}{2(\chi-\beta)}.\\
\end{aligned}
\end{equation}
\end{Claim}

{\it Proof.}
As shown in Sec.~\ref{Sec:chernoff}, we need to solve the following equations
\begin{equation} \label{symmteric:successful4}
\begin{aligned}
&2e^{-(\delta^L)^2{\bar \chi}/(2+\delta^L)}=\varepsilon,\\
&{\bar \chi}=\frac{\chi}{1+\delta^L}, \mbox{$0<\delta^L<1$},\\
\end{aligned}
\end{equation}
and
\begin{equation} \label{simplified:successful5}
\begin{aligned}
&2e^{-(\delta^U)^2{\bar \chi}/(2+\delta^U)}=\varepsilon,\\
&{\bar \chi}=\frac{\chi}{1-\delta^U}, \mbox{$0<\delta^U<1$},\\
\end{aligned}
\end{equation}
whose positive roots are obtained to be
\begin{equation} \label{Solve ELX3}
\begin{aligned}
&\delta^L=\frac{3\beta+\sqrt{8\beta\chi+\beta^2}}{2(\chi-\beta)},\\
&\delta^U=\frac{\sqrt{8\beta\chi+9\beta^2}-\beta}{2(\chi+\beta)}.\\
\end{aligned}
\end{equation}

In order to have $0<\delta^U,\delta^L< 1$, the value of $\chi$ should be larger than $6\beta$. One can in principle use the above equations for $\delta^L$ and $\delta^U$ to find the corresponding lower and upper bounds for $\bar \chi$. In Eq.~\eqref{confidence interval1}, we have used a symmetric form for the deviation parameter by choosing $\delta = \delta^L$ for both lower and upper bounds. This sysmmteric form would give us a slightly looser upper bound as it can be shown that $\delta^U$ is smaller than $\delta^L$. In the limit of $\chi \rightarrow \infty$, the above symmetric formulation would nevertheless give us the same asymptotic values as obtained in Claim A.2, which indicates that the two methodologies are more or less the same for large values of $\chi$.




\section{Random sampling}\label{APP:theta}

Here, we review the standard random sampling method used for the phase error rate estimation \cite{fung2010practical}. Suppose there are $n_x+n_z$ qubits (or basis-independent quantum states) in total. Alice and Bob randomly pick $n_x$ qubits, measured in the $X$ basis, and obtain a bit error rate of $e^{bx}$. They need to estimate the phase error rate, $e^{pz}$, for the remaining $n_z$ qubits measured in the $Z$ basis. When the data size is infinite, for basis-independent states, $e^{pz}=e^{bx}$. When statistical fluctuations are taken into account, a deviation $\theta$ is expected between the two error rates. According to the random sampling analysis, the (failure) probability for $e^{pz}\ge e^{bx}+\theta$ is given by \cite{fung2010practical}
\begin{equation} \label{Sampling:epsilon}
\begin{aligned}
Pr(e^{pz}\ge e^{bx}+\theta)\le \frac{\sqrt{n_x+n_z}}{\sqrt{e^{bx}(1-e^{bx})n_xn_z}}2^{-(n_x+n_z)\xi(\theta)},\\
\end{aligned}
\end{equation}
where $\xi(\theta)=h(e^{bx}+\theta-q^x\theta)-q^xh(e^{bx})-(1-q^x)h(e^{bx}+\theta)$ and $q^x=n_x/(n_x+n_z)$. For a given failure probability $\varepsilon$, one can then numerically find $\theta$ that satisfies

\begin{equation} \label{Sampling:epsilon:num}
\begin{aligned}
\varepsilon = \frac{\sqrt{n_x+n_z}}{\sqrt{e^{bx}(1-e^{bx})n_xn_z}}2^{-(n_x+n_z)\xi(\theta)}.\\
\end{aligned}
\end{equation}

In the decoy-state scheme considered here, we can use the above random sampling method to upper bound $\theta$, by using the following substitutions
\begin{equation} \label{Sampling:decoy}
\begin{aligned}
 e^{bx} &\rightarrow  e_1^{bxU} \\
 e^{pz} &\rightarrow e_1^{psz}  \\
 n_x &\rightarrow M_{1}^{xL}  \\
 n_z &\rightarrow M_{1}^{zsL}  \\
\end{aligned}
\end{equation}
in Eq.~\eqref{Sampling:epsilon:num}. The upper bound of the phase error rate $e_1^{psz}$ is then given by
\begin{equation} \label{upper bound Random sampling}
\begin{aligned}
e_1^{pszU}=e_1^{bxU}+\theta.
\end{aligned}
\end{equation}
Note that in order to estimate the phase error rate in the $Z$-basis signal states, we can use all the data points in the $X$ basis. That is why we use $M_{1}^{xL}$ rather than $M_{1}^{xsL}$ in Eq.~\eqref{Sampling:decoy}.

\section{Gaussian analysis} \label{App:StdErrAna}
Here, we summarize the Gaussian analysis method in Ref.~\cite{XMA:PRA:2005,ma2012statistical}, where the quantum channel is assumed to fluctuate according to a Gaussian distribution. According to the central limit theorem, a lower bound of  $y_1$, an upper bound of $e_1y_1$ and hence, an upper bound of $e_1$ can be obtained by
\begin{equation} \label{QKD-ratenew4}
\begin{aligned}
\min{~y_1}, \;\;\; s.t.,\\
(1-\frac{n_\alpha}{\sqrt{M_{a}}})Q^{a}&\le e^{-\mu_a}Y_0+\dots+e^{-\mu_a}\frac{(\mu_a)^i}{i!}Y_i +\dots\le (1+\frac{n_\alpha}{\sqrt{M^{a}}})Q^{a},\\
a\in\{s,w,v\}.
\end{aligned}
\end{equation}
\begin{equation} \label{QKD-ratenew5}
\begin{aligned}
\max{~e_1y_1}, \;\;\; s.t.,\\
(1-\frac{n_\alpha}{\sqrt{E^{a}M^{a}}})E^{a}Q^{a}&\le e^{-\mu_a}e_0Y_0+\dots+e^  {-\mu_a}\frac{(\mu_a)^i}{i!}e_iY_i +\dots\le (1+\frac{n_\alpha}{\sqrt{E^{a}M^{a}}})E^{a}Q^{a},\\
a\in\{s,w,v\}.
\end{aligned}
\end{equation}
The number of standard deviation $n_\alpha$ in Eq.~\eqref{QKD-ratenew4} is directly related to the failure probability,
\begin{equation} \label{deviation}
\begin{aligned}
1- \mathbf{erf}(n_\alpha/\sqrt{2})=\varepsilon,
\end{aligned}
\end{equation}
where $\mathbf{erf}(x)=\frac{2}{\sqrt{\pi}}\int_{0}^{x}e^{-t^2}dt$ is the error function \cite{cody1993algorithm}.

\section{Chernoff+Hoeffding method} \label{App:Curty}

In \cite{curty2014finite}, the parameter ${\bar \chi}$ is estimated by Chernoff+Hoeffding method. While in our method we use the Chernoff bound for all positive values of $\chi$, in \cite{curty2014finite}, the authors use the Hoeffding inequality when the date size is small. In this section, we denote $\mu$ to be ${\bar \chi}$. Then $\chi$ can be written as $\mu+\delta$, where $\delta\in[-\Delta,\widehat{\Delta}]$. The parameters $\varepsilon_1$, $\varepsilon_2$ and $\varepsilon_3$ are, respectively, the failure probabilities of the lower bound with the Hoeffding inequality, the lower bound estimation of the Chernoff bound, and the upper bound estimation of the Chernoff bound.

First, a general lower bound $\mu^L$ is given according to the Hoeffding inequality.
\begin{equation} \label{Hoeffding inequality}
\begin{aligned}
\mu^L=\chi-\sqrt{n\ln(1/\varepsilon_1)/2},\\
\end{aligned}
\end{equation}
where $n$ is the total number of random variables $\chi_i$ and $\chi=\sum^{n}_{i=1}\chi_i$.
This lower bound is used to determine the estimated means of the Chernoff+Hoeffding method.

With the upper bound $\mu^L$ in Eq.~\eqref{Hoeffding inequality}, the following three tests are performed:
\begin{enumerate}
\item{test1:}
    $(2\varepsilon_2^{-1})^{1/\mu^L}\le e^{(4/4\sqrt{2})^2}$
\item{test2:}
$(\varepsilon_3^{-1})^{1/\mu^L}< e^{1/3}$
\item{test3:}
$((\varepsilon_3)^{1/\mu^L})< e^{[(2e-1)/2]^2}$
\end{enumerate}

According to the results of these tests, the upper bound and lower bound are estimated with different means. If a test is fulfilled, the according bound can be calculated with Chernoff bound, which gives a tighter estimation. When no tests is fulfilled, the according bound have to be calculated by the looser Hoeffding inequality.

When estimating the upper bound, we denote that $\mu^U=\chi+\Delta$. According to the result of test$1$, the value of $\Delta$ is given by,
\begin{enumerate}
\item
when test1 is fulfilled, $\Delta=g(\chi,\varepsilon_2^4/16)$, where $g(x,y)=\sqrt{2x\ln(y^{-1})}$;

\item
when test1 is not fulfilled, $\Delta=\sqrt{n/2\ln(1/\varepsilon_2)}$
\end{enumerate}

When considering the lower bound, we denote that $\mu^L=\chi-\widehat{\Delta}$. According to the results of test$2$ and test$3$, the value of $\widehat{\Delta}$ is given by
\begin{enumerate}
\item
When test$2$ is fulfilled, $\widehat{\Delta}=g(\chi,\varepsilon_3^{3/2})$
\item
When test$2$ is not fulfilled, but test3 is fulfilled, $\widehat{\Delta}=g(\chi,\varepsilon_3^2)$
\item
When test$3$ is not fulfilled (test2 is also not fulfilled), $\widehat{\Delta}=\sqrt{n/2\ln(1/\varepsilon_3)}$
\end{enumerate}
\begin{corollary} \label{corollary:Hoeffding chernof:bound}
When all of the tests are fulfilled, $\varepsilon_3=\varepsilon_2=\varepsilon/2$, and $\chi\rightarrow\infty$, the confidence interval of $\bar \chi$ in Eq.~\eqref{confidence interval1} is given by,

\begin{equation} \label{confidence Hoeffding interval1new}
\begin{aligned}
&{\mathbb{E}}^L[\chi]=\chi(1 -\sqrt{\frac{3\beta}{\chi}}), {\mathbb{E}}^U[\chi]=\chi(1 +2\sqrt{\frac{2\beta-\ln2}{\chi}}).\\
\end{aligned}
\end{equation}
\end{corollary}
{\it Proof.}\\
When all of the tests are fulfilled, we know that:
\begin{equation} \label{confidence Hoeffding interval1new1}
\begin{aligned}
&{\bar \chi}^L(\chi)=\chi-g(\chi,\varepsilon_3^{3/2})=\chi(1 -\sqrt{\frac{3\beta}{\chi}}),\\
&{\bar \chi}^U(\chi)=\chi+g(\chi,\varepsilon_2^4/16)=\chi(1 +2\sqrt{\frac{2\beta-\ln2}{\chi}}).\\
\end{aligned}
\end{equation}

\bibliographystyle{apsrev4-1}

\bibliography{BibChernoff}


\end{document}